\documentclass[10pt,letterpaper,conference,twocolumn]{IEEEtran}%
\usepackage{graphicx}
\usepackage{amsmath}
\usepackage{amssymb}
\usepackage[colorlinks]{hyperref}
\usepackage{amsfonts}
\usepackage{psfrag}%
\providecommand{\U}[1]{\protect\rule{.1in}{.1in}}
\newtheorem{theorem}{Theorem}
\newtheorem{lemma}{Lemma}
\begin{document}

\title{Analysis of Alternative Metrics for the PAPR Problem in OFDM Transmission}
\author{Gerhard Wunder\\Fraunhofer Heinrich Hertz Institut\\ Berlin, Germany \\{\small Email: gerhard.wunder@hhi.fraunhofer.de}}
\maketitle

\begin{abstract}
The effective PAPR of the transmit signal is the standard metric to capture
the effect of nonlinear distortion in OFDM transmission. A common rule of
thumb is the log$\left(  N\right)  $ barrier where $N$ is the number of
subcarriers which has been theoretically analyzed by many authors. Recently,
new alternative metrics have been proposed in practice leading potentially to
different system design rules which are theoretically analyzed in this paper.
One of the main findings is that, most surprisingly, the log$\left(  N\right)
$ barrier turns out to be much too conservative: e.g. for the so-called
amplifier-oriented metric the scaling is rather $\log\left[  \log\left(
N\right)  \right]  $. To prove this result, new upper bounds on the PAPR
distribution for coded systems are presented as well as a theorem relating
PAPR results to these alternative metrics.

\end{abstract}

\section{Introduction}

The peak-to-average power ratio (PAPR) problem is a well-established problem
in OFDM literature and has entailed numerous research papers since the mid
nineties \cite{litsyn_07}. Nowadays, even though OFDM has become the
predominant wireless technology in the downlink, there are still many concerns
about the application of OFDM in the uplink. This is mainly due to the fact
that the PAPR reduction capabilities of state-of-the-art algorithms and their
respective impact on relevant performance measures such as power efficiency,
error probability, and spectral regrowth are not easy to track and mostly
presented in terms of simulations. This situation is indeed dissatisfactory
for system design, where provable performance limits are required. Another
important driving factor within the context of \emph{Green Information
Technology} is the growing energy cost of network operation setting standards
beyond capabilities of current PAPR reduction algorithms \cite{correia_10}.
Hence, it becomes more and more apparent that the problem can not be
considered as solved yet and that the PAPR metric itself has to be carefully
reviewed overthrowing some of the common understanding and results
particularly in the context of MIMO
\cite{fischer_09_inf,siegl_2010,pohl_07_ett}.

This paper revisits the PAPR problem and analyzes new performance metrics in
terms of their effective behaviour. Standard results suggest that power
amplifier backoff is to be adjusted along the $\log\left(  N\right)  $ rule of
thump where $N$ is the number of subcarriers
\cite{dinur_01_comm,ochiai_01_comm}. However, high but very narrow peaks
obviously cause spectral regrowth but the effect on e.g. symbol error
probability might be negligible which suggests that amplifier backoff and PAPR
of the transmit signal can indeed fall apart while, still, zero symbol error
probability can be achieved. This motivates the analysis of alternative
performance measures recently proposed in practice \cite{3gpp_04}. The
detailed contributions are as follows:

\textbf{Contributions:} First, we provide a new analytical upper bound on the
PAPR distribution for coded OFDM systems generalizing some known results in
the literature. The theorems are used to bound some given alternative
performance metric introducing the so-called balancing method. In this context
we prove that even though PAPR is of order $\log\left(  N\right)  $ with high
probability, the amplifier backoff can be adjusted according to a much lower
value. Specifically, for the so-called amplifier oriented metric the scaling
turns out to be of order $\log\left[  \log\left(  N\right)  \right]  $ which
is almost a constant in practical terms and suggests new system design rules.

\section{The Communication Model}

Let us introduce coded OFDM systems. We adopt the system model introduced by
\cite{paterson_00_inf2}: Let $\mathcal{C}$ be a code that maps $k_{b}$ input
bits into blocks of $N$ constellation symbols $c_{0},\ldots,c_{N-1}$, from a
complex constellation $\mathcal{Q}$ forming the codeword $\mathbf{c}$. We
assume here $\mathcal{Q}:=\{-1,1\}=$ \textit{BPSK}. The rate $R$ of this code
is defined to be $R=k_{b}/N$ such that $\mathcal{C}$ has $M_{1}=2^{RN}$ codewords.

Given a codeword $\mathbf{c}$, a single OFDM baseband symbol can be described
by
\begin{equation}
S_{\mathbf{c}}\left(  t\right)  =\frac{1}{\sqrt{N}}\sum_{k=0}^{N-1}%
c_{k}e^{2\pi jk\Delta ft},\quad0\leq t\leq T_{s},j=\sqrt{-1}
\label{eqn:ofdm_symbol}%
\end{equation}
where $N$ is the number of subcarriers, $\Delta f=1/T_{s}$ is the subcarrier
frequency offset and $T_{s}$ is the symbol duration. For mathematical
convenience the time axis can be normalized by $T_{s}$, i.e. we substitute
$\theta\left(  t\right)  =2\pi t/T_{s}$ and write $S_{\mathbf{c}}\left(
\theta\right)  ,0\leq\theta\leq2\pi$. Furthermore, for later reference define
$S_{\mathbf{c}}\left(  \theta,\alpha\right)  :=\Re e\left(  S_{\mathbf{c}%
}\left(  \theta\right)  e^{j\alpha}\right)  $ with sampling points%
\[
\theta_{l,L}:=\frac{2\pi l}{2LN},\alpha_{l,K}:=\frac{l2\pi}{K},
\]
which are collected in the two-dimensional lattice%
\[
\Omega_{L,K}:=\left\{  \left(  \theta_{l_{1},L},\alpha_{l_{2},K}\right)
,\;0\leq l_{1}<LN,0\leq l_{2}<K\right\}  ,
\]
of the square $\left[  0,2\pi\right)  \times\left[  0,2\pi\right)  $. Here,
$L>1$ is the oversampling factor and $K>2$ is some auxiliary variable. The
Nyquist-rate samples are $\theta_{l}:=\theta_{l,1}$. In the baseband model the
OFDM symbols undergo a nonlinear transformation denoted as%
\[
\Phi:S_{\mathbf{c}}\left(  \theta\right)  \hookrightarrow\Phi\left[
S_{\mathbf{c}}\left(  \theta\right)  \right]
\]
representing some high power amplifier (HPA) model. In the sequel, we assume
for simplicity that the nonlinearity acts solely on the samples obtained with
some oversampling factor $L$.

\section{HPA Models}

\subsection{Soft envelope limiter model}

The soft envelope limiter (SEL) model is given by%
\[
\Phi_{sel}\left(  S_{\mathbf{c}}\left(  \theta\right)  \right)  =\left\{
\begin{array}
[c]{ll}%
S_{\mathbf{c}}\left(  \theta\right)  , & \left\vert S_{\mathbf{c}}\left(
\theta\right)  \right\vert \leq\lambda\\
\lambda e^{j\arg\left(  S_{\mathbf{c}}\left(  \theta\right)  \right)  }, &
\left\vert S_{\mathbf{c}}\left(  \theta\right)  \right\vert >\lambda
\end{array}
\right.  ,
\]
where $\lambda$ is the saturation level of the non-linearity and the event
$\left\{  \left\vert S_{\mathbf{c}}\left(  \theta\right)  \right\vert
>\lambda\right\}  $ is commonly described as clipping. The samples after the
SEL nonlinearity can be decomposed as%
\[
\Phi\left(  S_{\mathbf{c}}\left(  \theta_{l,L}\right)  \right)  =S_{\mathbf{c}%
}\left(  \theta_{l,L}\right)  +D_{\mathbf{c}}\left(  \theta_{l,L}\right)
\]
and obviously $D_{\mathbf{c}}\left(  \theta_{l,L}\right)  =0$ $\forall l$ when
no clipping occurs. The SEL model is a standard model when there are
additional predistortion techniques.

\subsection{Cubic polynomial model}

Particular in the 3GPP\ context \cite{3gpp_04,siegl_2010} the cubic model has
become popular and is given by%
\[
\Phi_{cu}\left(  S_{\mathbf{c}}\left(  \theta\right)  \right)  =a\cdot
S_{\mathbf{c}}\left(  \theta\right)  +b\cdot S_{\mathbf{c}}\left(
\theta\right)  \left\vert S_{\mathbf{c}}\left(  \theta\right)  \right\vert
^{2},\;a,b>0.
\]
The advantage is in most cases a simpler analytical treatment.

\section{Figures of Merit}

\subsection{Crest-factor}

The crest-factor\ (CF)\footnote{We consider CF instead of PAPR.} of eqn.
(\ref{eqn:ofdm_symbol}) is defined by
\begin{equation}
\text{\textit{CF}}_{L}\left(  S_{\mathbf{c}}\right)  :=\max_{0\leq
l<LN}\left\vert S_{\mathbf{c}}\left(  \theta_{l,L}\right)  \right\vert
\label{eqn:cf}%
\end{equation}
with $1\leq$\textit{CF}$_{L}\left(  S_{\mathbf{c}}\right)  \leq\sqrt{N}$. A
first choice to assess the impact of a nonlinearity in the transmitter path
would be the maximum CF taken over all codewords (that we call the CF of a
code) defined by
\[
\text{\textit{CF}}_{L}\left(  \mathcal{C}\right)  :=\max_{c\in\mathcal{C}%
}\text{\textit{CF}}_{L}\left(  S_{\mathbf{c}}\right)  .
\]
It was shown in \cite{wunder_02_itg,tarokh_00_comm} that for spherical codes
the CF of a code can be computed with arbitrary accuracy provided that the
code supports minimum-distance decoding. On the other hand, the occurrence of
this ,,worst-case\textquotedblleft\ codeword may be extremely unlikely. In
this case, the distribution of the CF must be taken into account. The
complementary cumulative distribution function (CCDF) of (\ref{eqn:cf}) is
defined by
\[
B_{L}\left(  x\right)  :=\Pr\left(  \left\{  \text{\textit{CF}}_{L}\left(
S_{\mathbf{c}}\right)  >x;c\in\mathcal{C}\right\}  \right)  .
\]
($\Pr$ denotes probability). Using the CCDF a more appropriate measure can be
defined such as the ,,effective\textquotedblleft\ CF defined by the CF of
which the probability of occurrence may be considered negligible in practice,
i.e.%
\[
B_{L}\left(  \text{\textit{CF}}_{eff}\left(  \mathcal{C}\right)  \right)
=\epsilon
\]
where \textit{CF}$_{eff}\left(  \mathcal{C}\right)  $ is the effective CF and
$\epsilon$ is some small number, say $10^{-3}...10^{-8}$ (outage probability).
Our aim is to bound this term for codes.

\subsection{Amplifier-oriented metric}

The definition of $D_{\mathbf{c}}\left(  \theta\right)  $ suggest the
following metric%
\[
\text{\textit{AOM}}_{L}\left(  S_{\mathbf{c}}\right)  :=\frac{1}{LN}\sum
_{l=0}^{LN-1}\left\vert D_{\mathbf{c}}\left(  \theta_{l,L}\right)  \right\vert
^{2}%
\]
In the following we will see how we can relate this metric to $B_{L}$.

\section{Fundamentals on CF Distribution}

The general approach is as follows: suppose the probabilities
\[
\Pr\left(  S_{\mathbf{c}}\left(  \theta,\alpha\right)  >x\right)
\]
are given where the tuple $\left(  \theta,\alpha\right)  $ runs through
$\Omega_{L,K}$. By the method of projections and union bound we have
\[
B_{L}\left(  x\right)  \leq\min_{K>2}\sum_{\left(  \theta,\alpha\right)
\in\Omega_{L,K}}\Pr\left(  S_{\mathbf{c}}\left(  \theta,\alpha\right)
>\frac{x}{C_{K}}\right)
\]
where $C_{K}:=\cos^{-1}\left(  \frac{\pi}{K}\right)  ,K\geq3$. Thus, all we
have to do is to bound the distribution of the instantaneous envelope for the
different modulation schemes. Thus, we replace the probability terms with the
Chernoff (or any other Marcov style) bound, i.e.
\begin{equation}
\Pr\left(  S_{\mathbf{c}}\left(  \theta,\alpha\right)  >x\right)
\leq\mathbb{E}\left(  e^{\varrho\left(  S_{\mathbf{c}}\left(  \theta
,\alpha\right)  -x\right)  }\right)  \label{eqn:chernoff}%
\end{equation}
for any (real) $\varrho>0$ ($\mathbb{E}\left(  \cdot\right)  $ is the
expectation operator). We call the bound (\ref{eqn:chernoff}) a union bound on
the CCDF of the CF.

We also need some elements from coding theory. Let $d\left(  \mathbf{c}%
^{\prime},\mathbf{c}^{\prime\prime}\right)  $ be the \emph{Hamming distance}
between codewords $\mathbf{c}^{\prime},\mathbf{c}^{\prime\prime}$, i.e. the
number of positions where the codewords differ. The distance distribution is
then given by
\[
W_{k}^{%
%TCIMACRO{\U{b0}}%
%BeginExpansion
{{}^\circ}%
%EndExpansion
}:=\frac{\left\vert \left\{  \mathbf{c}^{\prime},\mathbf{c}^{\prime\prime}%
\in\mathcal{C}:d\left(  \mathbf{c}^{\prime},\mathbf{c}^{\prime\prime}\right)
=k\right\}  \right\vert }{M_{1}}.
\]
Note that since for a linear code it does not matter which individual codeword
we pick when calculating the distance to another codeword, the distance
distribution coincides with the weight distribution $W_{k}$ for linear codes,
i.e.
\[
W_{k}:=\left\vert \left\{  \mathbf{c}:w\left(  \mathbf{c}\right)
=k,\mathbf{c}\in\mathcal{C}\right\}  \right\vert
\]
where $w\left(  \cdot\right)  $ denotes the weight of a codeword. Furthermore,
if the code contains the all-one codeword (i.e. all components are negative
under our identification of constellation symbols) we have $W_{k}=W_{N-k}$ and
the weight distribution becomes \emph{symmetric}.

The main purpose of the following derivations is to prove an interesting
connection between CF distributions and distance distributions. Since many
results require symmetric weight and distance distributions we start with the
following lemma. It says that the CF distribution can be estimated by its
symmetrized version.

\begin{lemma}
\label{lemma:help} Suppose $\mathcal{C}_{A}$ is a binary code. Then, for any
set $\mathcal{A}:=\left\{  \mathbf{c}\in\mathcal{C}_{A}:\text{\textit{CF}%
}\left(  S_{\mathbf{c}}\right)  >x\right\}  $ the probability that this occurs
is upperbounded by
\[
\frac{\left\vert \mathcal{A}\right\vert }{\left\vert \mathcal{C}%
_{A}\right\vert }\leq2\frac{\left\vert \mathcal{A}\cup\mathcal{B}\right\vert
}{\left\vert \mathcal{C}_{A}\cup\mathcal{C}_{B}\right\vert }\leq
4\frac{\left\vert \mathcal{A}\right\vert }{\left\vert \mathcal{C}%
_{A}\right\vert },
\]
where $\mathcal{C}_{B}$ is the binary code constructed by adding the all-one
codeword to any codeword of $\mathcal{C}_{A}$ and $\mathcal{B}:=\left\{
\mathbf{c}\in\mathcal{C}_{B}:\text{\textit{CF}}\left(  S_{\mathbf{c}}\right)
>x\right\}  $.
\end{lemma}

\begin{proof}
First, observe that $\left\vert \mathcal{C}_{A}\right\vert =\left\vert
\mathcal{C}_{B}\right\vert $ and $\left\vert \mathcal{A}\right\vert
=\left\vert \mathcal{B}\right\vert $. Furthermore we have
\[
\left\vert \mathcal{C}_{A}\cup\mathcal{C}_{B}\right\vert \leq2\left\vert
\mathcal{C}_{A}\right\vert
\]
so that we get%
\[
\left\vert \mathcal{C}_{A}\right\vert \geq\frac{\left\vert \mathcal{C}_{A}%
\cup\mathcal{C}_{B}\right\vert }{2}.
\]
Hence, we have
\[
\frac{\left\vert \mathcal{A}\right\vert }{\left\vert \mathcal{C}%
_{A}\right\vert }\leq\frac{\left\vert \mathcal{A}\cup\mathcal{B}\right\vert
}{\left\vert \mathcal{C}_{A}\right\vert }\leq2\frac{\left\vert \mathcal{A}%
\cup\mathcal{B}\right\vert }{\left\vert \mathcal{C}_{A}\cup\mathcal{C}%
_{B}\right\vert }.
\]
The converse is%
\[
\frac{\left\vert \mathcal{A}\cup\mathcal{B}\right\vert }{\left\vert
\mathcal{C}_{A}\cup\mathcal{C}_{B}\right\vert }\leq\frac{2\left\vert
\mathcal{A}\right\vert }{\left\vert \mathcal{C}_{A}\cup\mathcal{C}%
_{B}\right\vert }\leq\frac{2\left\vert \mathcal{A}\right\vert }{\left\vert
\mathcal{C}_{A}\right\vert }%
\]
proving the claim.
\end{proof}

Now, we are ready for our first theorem which generalizes \cite[Thm.
6.16]{litsyn_07}.

\begin{theorem}
\label{theorem:cdf_code_nl} For any binary code $\mathcal{C}$ the CCDF of the
CF is upperbounded by
\[
B_{L}\left(  x\right)  \leq\min_{\varrho>0}\sqrt{\sum_{k=0}^{N}\frac{f^{\ast
}\left(  \varrho,x\right)  W_{k}^{%
%TCIMACRO{\U{b0}}%
%BeginExpansion
{{}^\circ}%
%EndExpansion
}\cosh\left(  \varrho N^{\frac{1}{4}}\left(  N-2k\right)  \right)  }{M_{1}}}%
\]
where
\[
f^{\ast}\left(  \varrho,x\right)  :=\min_{K>2}2LNK\exp\left(  -\frac
{\varrho\sqrt{N}x}{C_{K}}\right)  \sqrt{\cosh\left(  \varrho N^{\frac{3}{4}%
}\right)  }.
\]

\end{theorem}

\begin{proof}
We first assume codes with symmetric weight distributions. For ease of
presentation let us define
\[
b_{\mathbf{i}}^{\left(  j\right)  }:=\frac{j!}{i_{0}!i_{1}!\cdots i_{N-1}!}%
\]
and
\[
k_{\mathbf{i}}\left(  \theta,\alpha\right)  :=\left(  \cos^{i_{0}}\left(
\alpha\right)  ,\ldots,\cos^{i_{N-1}}\left(  \left(  N-1\right)  \theta
+\alpha\right)  \right)  ,
\]
as well as the code moments%
\[
M_{\mathbf{i}}:=\mathbb{E}\left(  c_{0}^{i_{0}}c_{1}^{i_{1}}\ldots
c_{N-1}^{i_{N-1}}\right)  .
\]
Fixing $\varrho>0,L>1,K>2,N_{1}>1$, and expanding the exponential function in
the Chernoff bound yields
\[
B_{L}\left(  x\right)  \lesssim\sum_{\left(  \theta,\alpha\right)  \in
\Omega_{L,K}}e^{-\frac{\varrho\sqrt{N}x}{C_{K}}}\sum_{j=0}^{N_{1}}%
\frac{\varrho^{j}}{j!}\sum_{\mathbf{i}\in\mathcal{I}_{j}}b_{\mathbf{i}%
}^{\left(  j\right)  }k_{\mathbf{i}}\left(  \theta,\alpha\right)
M_{\mathbf{i}}%
\]
where we have omitted the error term (indicated by $\lesssim$) on the right
hand side which depends on the natural number $N_{1}>0$ and is given by
Taylor's theorem. Applying the Cauchy-Schwartz's inequality yields
\begin{align*}
B_{L}\left(  x\right)   &  \lesssim\sum_{\left(  \theta,\alpha\right)
\in\Omega_{L,K}}e^{-\frac{\varrho\sqrt{N}x}{C_{K}}}\\
&  \sum_{j=0}^{N_{1}}\frac{\varrho^{j}}{j!}\left(  \sum_{\mathbf{i}%
\in\mathcal{I}_{j}}b_{\mathbf{i}}^{\left(  j\right)  }k_{\mathbf{i}}%
^{2}\left(  \theta,\alpha\right)  \right)  ^{\frac{1}{2}}\left(
\sum_{\mathbf{i}\in\mathcal{I}_{j}}b_{\mathbf{i}}^{\left(  j\right)
}M_{\mathbf{i}}^{2}\right)  ^{\frac{1}{2}}.
\end{align*}
Observing that
\begin{align*}
\sum_{\mathbf{i}\in\mathcal{I}_{j}}b_{\mathbf{i}}^{\left(  j\right)
}k_{\mathbf{i}}^{2}\left(  \theta,\alpha\right)   &  =\left(  \sum_{k=0}%
^{N-1}\cos^{2}\left(  k\theta+\alpha\right)  \right)  ^{j}\\
&  =:\Psi_{N}^{j}\left(  \theta,\alpha\right)
\end{align*}
we have%
\begin{align*}
B_{L}\left(  x\right)   &  \lesssim\sum_{\left(  \theta,\alpha\right)
\in\Omega_{L,K}}e^{-\frac{\varrho\sqrt{N}x}{C_{K}}}\\
&  \sum_{j=0}^{N_{1}}\frac{\varrho^{j}\Psi_{N}^{\frac{j}{2}}\left(
\theta,\alpha\right)  }{j!}\left(  \sum_{\mathbf{i}\in\mathcal{I}_{j}%
}b_{\mathbf{i}}^{\left(  j\right)  }M_{\mathbf{i}}^{2}\right)  ^{\frac{1}{2}}.
\end{align*}
Next, we can replace the squared moments in terms of the distance distribution
of the (in general nonlinear) code \cite[Proof of Thm. 6.16]%
{litsyn_06_inf,litsyn_07}, i.e.%
\[
\sum_{\mathbf{i}\in\mathcal{I}_{j}}b_{\mathbf{i}}^{\left(  j\right)  }\left(
\sum_{\mathbf{c}\in\mathcal{C}}\prod_{k=0}^{N-1}c_{k}^{i_{k}}\right)
^{2}=M_{1}\sum_{k=0}^{N}\left(  N-2k\right)  ^{j}W_{k}^{%
%TCIMACRO{\U{b0}}%
%BeginExpansion
{{}^\circ}%
%EndExpansion
},
\]
and therefore%
\begin{align*}
B_{L}\left(  x\right)   &  \lesssim\sum_{\left(  \theta,\alpha\right)
\in\Omega_{L,K}}e^{-\frac{\varrho\sqrt{N}x}{C_{K}}}\\
&  \sum_{j=0}^{N_{1}}\frac{\varrho^{2j}\Psi_{N}^{j}\left(  \theta
,\alpha\right)  }{\left(  2j\right)  !}\left(  \frac{1}{M_{1}}\sum_{k=0}%
^{N}\left(  N-2k\right)  ^{2j}W_{k}^{%
%TCIMACRO{\U{b0}}%
%BeginExpansion
{{}^\circ}%
%EndExpansion
}\right)  ^{\frac{1}{2}}%
\end{align*}
since the distance distribution is symmetric. Again applying Cauchy-Schwartz's
inequality yields
\begin{align*}
&  \sum_{j=0}^{N_{1}}\frac{\varrho^{2j}\left[  \Psi_{N}\left(  \theta
,\alpha\right)  \right]  ^{\frac{3j}{4}}}{\left(  2j\right)  !}\left(
\frac{\left[  \Psi_{N}\left(  \theta,\alpha\right)  \right]  ^{\frac{j}{2}}%
}{M_{1}}\sum_{k=0}^{N}\left(  N-2k\right)  ^{2j}W_{k}^{%
%TCIMACRO{\U{b0}}%
%BeginExpansion
{{}^\circ}%
%EndExpansion
}\right)  ^{\frac{1}{2}}\\
&  \leq\left(  \sum_{j=0}^{N_{1}}\frac{\varrho^{2j}\left(  \left[  \Psi
_{N}\left(  \theta,\alpha\right)  \right]  ^{\frac{3}{4}}\right)  ^{2j}%
}{\left(  2j\right)  !}\right)  ^{\frac{1}{2}}\\
&  \left(  \sum_{j=0}^{N_{1}}\frac{\varrho^{2j}}{M_{1}\left(  2j\right)
!}\sum_{k=0}^{N}\left(  \left[  \Psi_{N}\left(  \theta,\alpha\right)  \right]
^{\frac{1}{4}}\right)  ^{2j}\left(  N-2k\right)  ^{2j}W_{k}^{%
%TCIMACRO{\U{b0}}%
%BeginExpansion
{{}^\circ}%
%EndExpansion
}\right)  ^{\frac{1}{2}}.
\end{align*}
Since $N_{1}$ is arbitrary, we have%
\[
\sum_{j=0}^{N_{1}}\frac{\varrho^{2j}\left(  \left[  \Psi_{N}\left(
\theta,\alpha\right)  \right]  ^{\frac{3}{4}}\right)  ^{2j}}{\left(
2j\right)  !}\underset{N_{1}\rightarrow\infty}{\rightarrow}\cosh\left(
\varrho\left[  \Psi_{N}\left(  \theta,\alpha\right)  \right]  ^{\frac{3}{4}%
}\right)  ,
\]
and%
\begin{align*}
&  \sum_{j=0}^{N_{1}}\frac{\varrho^{2j}}{M_{1}\left(  2j\right)  !}\sum
_{k=0}^{N}\left(  \left[  \Psi_{N}\left(  \theta,\alpha\right)  \right]
^{\frac{1}{4}}\right)  ^{2j}\left(  N-2k\right)  ^{2j}W_{k}^{%
%TCIMACRO{\U{b0}}%
%BeginExpansion
{{}^\circ}%
%EndExpansion
}\underset{N_{1}\rightarrow\infty}{\rightarrow}\\
&  \frac{1}{M_{1}}\sum_{k=0}^{N}\cosh\left(  \varrho\left[  \Psi_{N}\left(
\theta,\alpha\right)  \right]  ^{\frac{1}{4}}\left(  N-2k\right)  \right)
W_{k}^{%
%TCIMACRO{\U{b0}}%
%BeginExpansion
{{}^\circ}%
%EndExpansion
}%
\end{align*}
for any $\varrho>0,L>1,K>2$. The final result follows from $\Psi_{N}%
^{j}\left(  \theta,\alpha\right)  \leq N^{j}$ and invoking Lemma
\ref{lemma:help} to lift the proof to the non-symmetric case. The additional
factor $\frac{1}{2}$ is due to the real BPSK symbols where $|S_{\mathbf{c}%
}\left(  \theta\right)  |=|S_{\mathbf{c}}\left(  2\pi-\theta\right)  |$.
\end{proof}

We can improve on the result by assuming linearity of the code. The following
theorem relies on the fact that moments of linear codes are non-negative which
generalizes \cite[Thm. 6.13]{litsyn_06_inf,litsyn_07}.

\begin{theorem}
\label{theorem:cdf_code_lin} Let $\mathcal{C}$ be a linear, binary code. Then
the CCDF of the CF is upperbounded by
\[
B_{L}\left(  x\right)  \leq\min_{\varrho>0}\sum_{k=0}^{N}\frac{f^{\ast\ast
}\left(  \varrho,x\right)  W_{k}\cosh\left(  \varrho\left(  N-2k\right)
\right)  }{M_{1}}%
\]
where%
\[
f^{\ast\ast}\left(  \varrho,x\right)  :=\min_{K>2}2LNK\exp\left(
-\frac{\varrho\sqrt{N}x}{C_{K}}\right)  .
\]

\end{theorem}

\begin{proof}
The proof uses Lemma \ref{lemma:help} and invokes the same proof steps as in
\cite{litsyn_06_inf} which are omitted.
\end{proof}

The applicability of the latter theorems is ensured by the following result.

\begin{theorem}
\label{theorem:cdf_codes} Suppose that for all $W_{k}^{%
%TCIMACRO{\U{b0}}%
%BeginExpansion
{{}^\circ}%
%EndExpansion
}$ there is a constant $C_{w}$ independent of $k$ so that%
\begin{equation}
W_{k}^{%
%TCIMACRO{\U{b0}}%
%BeginExpansion
{{}^\circ}%
%EndExpansion
}\leq\left(  1+C_{w}\right)  \frac{1}{2^{N-k_{b}}}\binom{N}{k}%
.\label{eqn:w_cond}%
\end{equation}
Then the CF is upperbounded by:%
\[%
\begin{array}
[c]{ll}%
B_{L}\left(  x\right)  \leq2\left(  1+C_{w}\right)  LKN\exp\left(
-\frac{x^{2}}{2C_{K}^{2}\sqrt{N}}\right)   & \text{nonlinear }\mathcal{C}%
\text{ }\\
B_{L}\left(  x\right)  \leq2\left(  1+C_{w}\right)  LKN\exp\left(
-\frac{x^{2}}{2C_{K}^{2}}\right)   & \text{linear }\mathcal{C}%
\end{array}
\]

\end{theorem}

\begin{proof}
We omit the proof for the linear part which follows directly from from
\cite{litsyn_06_inf}. In the nonlinear case by Theorem
\ref{theorem:cdf_code_nl} and by virtue of%
\begin{align*}
\sqrt{\sum_{k=0}^{N}\binom{N}{k}\frac{\cosh\left(  \varrho N^{\frac{1}{4}%
}\left(  N-2k\right)  \right)  }{2^{N}}} &  \leq\left(  \cosh\left(
N^{\frac{1}{4}}\varrho\right)  \right)  ^{\frac{N}{2}}\\
&  \leq\exp\left(  \frac{\varrho^{2}N^{\frac{3}{2}}}{4}\right)  ,
\end{align*}
and since%
\[
\sqrt{\cosh\left(  \varrho N^{\frac{3}{4}}\right)  }\leq\exp\left(
\frac{\varrho^{2}N^{\frac{3}{2}}}{4}\right)  ,
\]
we obtain%
\[
B_{L}\left(  x\right)  \leq2\left(  1+C_{w}\right)  LNK\exp\left(
-\frac{\varrho x\sqrt{N}}{C_{K}}\right)  \exp\left(  \frac{\varrho^{2}%
N^{\frac{3}{2}}}{4}\right)  .
\]
Setting $\varrho=\frac{2x}{C_{K}N}$ yields%
\[
\bar{F}_{cf}^{\mathcal{C}}\left(  x\right)  \leq2\left(  1+C_{w}\right)
LNK\exp\left(  -\frac{x^{2}}{C_{K}^{2}\sqrt{N}}\right)  .
\]
The asymptotics follow immediately then.
\end{proof}

\section{Alternative Metrics:\ Upper Bounds}

In this section we use Theorems \ref{theorem:cdf_code_nl},
\ref{theorem:cdf_code_lin} to obtain upper bounds for the alternative metrics.
Let us define the random variable%
\begin{equation}
N_{\mathbf{c}}^{\left(  L\right)  }\left(  \lambda\right)  :=\left\vert
\left\{  l:\left\vert S_{\mathbf{c}}\left(  \theta_{l,L}\right)  \right\vert
>\lambda,l=0,\ldots,LN-1\right\}  \right\vert
\end{equation}
counting the number of samples which exceed a level $\lambda$. In the
following theorem we apply a balancing technique between the CCDF of the CF
and the number of samples that exceed a given level. For ease of presentation
we set $L=1$.

We start with the uncoded BPSK case where a tighter bound can be obtained.
Note that BPSK is a linear code with $C_{w}=0$ in Theorem
\ref{theorem:cdf_codes}.

\begin{theorem}
\label{theorem:ser_bound} Assume $\mathcal{Q}:=\{-1,1\}=$ \textit{BPSK} and
suppose $S_{\mathbf{c}}$ is clipped at the level $\lambda$. For any given
performance metric $h$ which is increasing in its argument, the average
distortion is upperbounded by
\begin{align*}
&  \Pr\left(  \sum_{l=0}^{N-1}h\left(  \left\vert D_{\mathbf{c}}\left(
\theta_{l}\right)  \right\vert \right)  >x\right) \\
&  \leq\min_{\mu>\lambda}\left[  B_{1}\left(  \mu\right)  +\left[
B_{1}\left(  \lambda\right)  +B_{1}^{2}\left(  \lambda\right)  \right]
\frac{h^{2}\left(  \mu-\lambda\right)  }{x^{2}}\right]  .
\end{align*}

\end{theorem}

\begin{proof}
Define the event
\[
\mathcal{A}:=\left\{  \sum_{l=0}^{LN-1}h\left(  \left\vert D_{\mathbf{c}%
}\left(  \theta_{l}\right)  \right\vert \right)  >x\right\}  .
\]
$\mathcal{A}$ can be partioned into disjoint events $\mathcal{A}\cap\left\{
\text{\textit{CF}}_{L}\left(  S_{\mathbf{c}}\right)  >\mu\right\}  $ or
$\mathcal{A}\cap\left\{  \text{\textit{CF}}_{L}\left(  S_{\mathbf{c}}\right)
\leq\mu\right\}  $. Thus
\begin{equation}
\Pr\left(  \mathcal{A}\right)  \leq B_{L}^{\ast}\left(  \mu\right)
+\Pr\left(  \mathcal{A}\cap\left\{  \text{\textit{CF}}_{L}\left(
S_{\mathbf{c}}\right)  \leq\mu\right\}  \right)  .
\end{equation}
Next we need to calculate the term $\Pr\left(  \mathcal{A}\cap\left\{
\text{\textit{CF}}_{L}\left(  S_{\mathbf{c}}\right)  \leq\mu\right\}  \right)
$. Clearly the event $\mathcal{A}\cap\left\{  \text{\textit{CF}}_{L}\left(
S_{\mathbf{c}}\right)  \leq\mu\right\}  $ is contained in the event
\[
\left\{  \left\{  N_{\mathbf{c}}^{\left(  L\right)  }\left(  \lambda\right)
\cdot\max_{0\leq l<LN}h\left(  \left\vert D_{\mathbf{c}}\left(  \theta
_{l,L}\right)  \right\vert \right)  \geq x\right\}  \cap\left\{
\text{\textit{CF}}_{L}\left(  S_{\mathbf{c}}\right)  \leq\mu\right\}
\right\}
\]
which itself is within the event%
\[
\left\{  N_{\mathbf{c}}^{\left(  L\right)  }\left(  \lambda\right)  \cdot
h\left(  \mu-\lambda\right)  \geq x\right\}  .
\]
Hence, we take the unconstrained number of points exceeding $\lambda$ but lift
up the level that is needed for a countable event. Writing%
\begin{align*}
N_{\mathbf{c}}\left(  \lambda,L\right)   &  =\sum_{l=0}^{LN-1}\mathbb{I}%
\left\{  \left\vert S_{\mathbf{c}}\left(  \theta_{l,L}\right)  \right\vert
>\lambda\right\}  \\
&  \leq\sum_{\left(  \theta,\alpha\right)  \in\Omega_{L,K}}\mathbb{I}\left\{
S_{\mathbf{c}}\left(  \theta,\alpha\right)  >\frac{\lambda}{C_{K}}\right\}
\end{align*}
and by Markov's inequality applied to the squared term and evaluating the
exponential moments thereby using the inherent structure of $S_{\mathbf{c}%
}\left(  \theta,\alpha\right)  $ \cite{wunder_03_inf} yields%
\begin{align*}
&  \Pr\left(  N_{\mathbf{c}}^{2}\left(  \lambda,L\right)  >x^{2}\right)  \\
&  \leq\frac{NKe^{\frac{\varrho^{2}N}{2}-\frac{\varrho\lambda\sqrt{N}}{C_{K}}%
}+N\left(  N-1\right)  K^{2}e^{\frac{2\varrho^{2}N}{2}-\frac{2\varrho
\lambda\sqrt{N}}{C_{K}}}}{x^{2}}.
\end{align*}
Due to lack of space we omit the details. \begin{figure}[th]
\begin{center}
\psfrag{lambda}{$\lambda$}\psfrag{mu}{$\mu$}
\includegraphics[width=6cm]{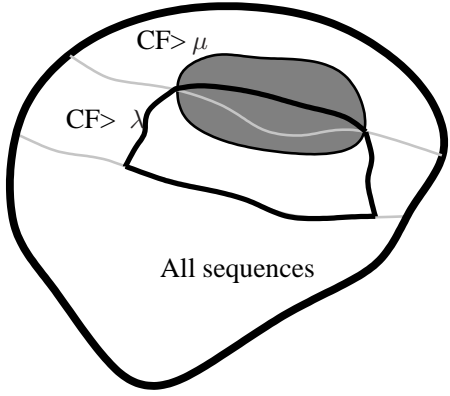}
\end{center}
\caption{Illustration of bounding method: the shaded area is the set of which
the probability measure is to be bounded. The upper bound is by adjusting the
$\mu$ level. The bold encircled area is the set where a sufficient number of
points crosses the $\lambda$ level.}%
\end{figure}
\end{proof}

The following theorem holds for any binary code.

\begin{theorem}
\label{theorem:ser_bound_general} Suppose $S_{\mathbf{c}}$ is clipped at the
level $\lambda$. For any given performance metric $h$ which is increasing in
its argument, the average distortion is upperbounded by
\begin{align*}
& \Pr\left(  \sum_{l=0}^{N-1}h\left(  \left\vert D_{\mathbf{c}}\left(
\theta_{l}\right)  \right\vert \right)  >x\right)  \\
& \leq\min_{\mu>\lambda}\left[  B_{1}\left(  \mu\right)  +B_{1}\left(
\lambda\right)  \frac{h\left(  \mu-\lambda\right)  }{x}\right]  .
\end{align*}

\end{theorem}

\section{Applications}

In this section we relate the theoretical results obtained so far to get
scaling result for alternative metrics. Suppose that we use a binary code and
that condition (\ref{eqn:w_cond}) is satisfied. Note that the condition holds
for uncoded transmission as well as many families of codes (such as the BCH
familiy \cite[pp. 158 ff.]{litsyn_07}). Suppose further we apply the AOM
metric in case of the SEL amplifier model. Setting
\[
h\left(  \left\vert D_{\mathbf{c}}\left(  \theta_{l}\right)  \right\vert
\right)  :=\frac{\left\vert D_{\mathbf{c}}\left(  \theta_{l}\right)
\right\vert ^{2}}{N}%
\]
then by Theorem \ref{theorem:ser_bound_general} we have (omitting constants
which have no effect on the asymptotic results)%
\begin{align*}
&  \Pr\left(  \sum_{l=0}^{LN-1}h\left(  \left\vert D_{\mathbf{c}}\left(
\theta_{l,L}\right)  \right\vert \right)  >x\right) \\
&  \leq\min_{\mu>\lambda}\left[  B_{1}\left(  \mu\right)  +B_{1}\left(
\lambda\right)  \frac{\left(  \mu-\lambda\right)  ^{2}}{N\cdot x}\right] \\
&  \leq\min_{\mu>\lambda}\left[  NKe^{-\frac{\mu^{2}}{2C_{K}^{2}}}%
+NKe^{-\frac{\lambda^{2}}{2C_{K}^{2}}}\frac{\left(  \mu-\lambda\right)  ^{2}%
}{N\cdot x}\right]  ,
\end{align*}
and setting
\begin{align*}
\lambda &  =\lambda_{N}=\sqrt{\left(  1+\varepsilon\right)  \log\left[
\log\left(  N\right)  \right]  }\\
\mu &  =\mu_{N}=\sqrt{\left(  1+\varepsilon\right)  \log\left(  N\right)  }%
\end{align*}
where $\varepsilon>0$ is a arbitralily small constant yields%
\[
\Pr\left(  \sum_{l=0}^{LN-1}\frac{\left\vert D_{\mathbf{c}}\left(
\theta_{l,L}\right)  \right\vert ^{2}}{LN}>x\right)  \rightarrow
0,\;N\rightarrow\infty
\]
for any fixed $x>0$. Hence we have the remarkable result that the clipping
level of the amplifier can be almost set to a constant and still the AOM
metric can be made arbitralily small.

\section{Conclusions}

In this paper we show that the standard design rule for the power amplifer in
OFDM transmission might be too conservative if alternative metrics are
considered. It is worth emphasizing that we do not claim that this is the
ultimate scaling as in practice other metrics might be important (e.g.
spectral regrowth). We provided an example and considered the amplifier
oriented metric recently proposed in practice. We have not yet considered new
algorithms (e.g. based on derandomization) for these metrics which is an
interesting extension of the results in this paper.

\bibliographystyle{IEEEbib}

\begin{thebibliography}{10}

\bibitem{litsyn_07}
S.~Litsyn,
\newblock {\em Peak Power Control in Multicarrier Communications},
\newblock Cambridge University Press, 2007.

\bibitem{correia_10}
{L. M. Correia, D. Zeller, O. Blume, D. Ferling, Y. Jading, I. Gódor, G. Auer,
  L. Van der Perre},
\newblock ``Challenges and enabling technologies for energy aware mobile radio
  networks,''
\newblock {\em IEEE Signal Processing Magazine}, vol. 48, no. 11, pp. 66--72,
  November 2010.

\bibitem{fischer_09_inf}
R.~F.~H. Fischer and C.~Siegl,
\newblock ``Reed-solomon and simplex codes for {PAPR} reduction in {OFDM},''
\newblock {\em IEEE Trans. Inform. Theory}, vol. 55, no. 4, pp. 1519--1528,
  April 2009.

\bibitem{siegl_2010}
C.~Siegl,
\newblock {\em Peak-to-average Power Ratio Reduction in Multi-antenna {OFDM}
  Via Multiple Signal Representation}, vol.~29,
\newblock Dissertation, Erlanger Berichte aus Informations- und
  Kommunikationstechnik, Shaker Verlag, Aachen, 2010.

\bibitem{pohl_07_ett}
H.~Boche and V.~Pohl,
\newblock ``Signal representation and approximation - fundamental limits,''
\newblock {\em European Trans. on Telecomm. (ETT)}, vol. 18, no. 5, pp.
  445--456, August 2007.

\bibitem{dinur_01_comm}
N.~Dinur and D.~Wulich,
\newblock ``Peak-to average power ratio in high-order {OFDM},''
\newblock {\em IEEE Trans. on Communications}, vol. 49, no. 6, pp. 1063--1072,
  June 2001.

\bibitem{ochiai_01_comm}
H.~Ochiai and H.~Imai,
\newblock ``On the distribution of the peak-to-average power ratio in {OFDM}
  signals,''
\newblock {\em IEEE Trans. on Communications}, vol. 49, no. 2, pp. 282--289,
  February 2001.

\bibitem{3gpp_04}
3GPP TSG RAN~WG1 37,
\newblock ``Comparison of {PAR} and cubic metric for power de-rating,''
\newblock in {\em Motorola Tdoc R1-040642}, May 2004.

\bibitem{paterson_00_inf2}
K.~G. Paterson and V.~Tarokh,
\newblock ``On the existence and construction of good codes with low
  peak-to-average power ratios,''
\newblock {\em IEEE Trans. Inform. Theory}, vol. 46, no. 6, pp. 1974--1986,
  September 2000.

\bibitem{wunder_02_itg}
G.~Wunder and H.~Boche,
\newblock ``A baseband model for computing the {PAPR} in {OFDM} systems,''
\newblock in {\em 4th International ITG Conference on Source and Channel
  Coding}, Berlin, January 2002, VDE, pp. 273--280, VDE-Verlag GmbH.

\bibitem{tarokh_00_comm}
V.~Tarokh and H.~Jafarkhani,
\newblock ``On the computation and reduction of the peak-to average power ratio
  in multicarrier communications,''
\newblock {\em IEEE Trans. on Comm.}, vol. 48, no. 1, pp. 37--44, January 2000.

\bibitem{litsyn_06_inf}
S.~Litsyn and G.~Wunder,
\newblock ``Generalized bounds on the {CF} distribution of {OFDM} signals with
  application to code design,''
\newblock {\em IEEE Trans. Inform. Theory}, vol. 52, no. 3, pp. 992--1006,
  March 2006.

\bibitem{wunder_03_inf}
G.~Wunder and H.~Boche,
\newblock ``New results on the statistical distribution of the crest-factor of
  {OFDM} signals,''
\newblock {\em IEEE Trans. Inform. Theory}, vol. 49, no. 2, pp. 488--494,
  February 2003.

\end{thebibliography}

\end{document}